\documentclass[11pt]{article}

\usepackage{amsmath, amssymb, amsthm}
\usepackage{geometry}
\usepackage{graphicx}
\newtheorem{theorem}{Theorem}
\newtheorem{remark}{Remark}
\newtheorem{corollary}{Corollary}
\newtheorem{definition}{Definition}
\title{Unified Operator Framework for Functional and Multivariate Regression}
\author{
Mark Carpenter and Nicholas Gaubatz\\[0.5em]
Department of Mathematics and Statistics\\
Auburn University\\
Auburn, AL
}
\date{}

\begin{document}

\maketitle

\begin{abstract}
We develop a unified operator framework for scalar, multivariate, and functional regression based on integral operators defined with respect to general measures. Within this framework, classical regression models, including scalar-on-function, function-on-scalar, function-on-function, and multivariate multiple regression, arise as special cases corresponding to different choices of input and output measures. We establish three main results. First, we show that the standard regression taxonomy can be expressed as a single operator under varying measures. Second, we demonstrate that discrete representations correspond to exact operator evaluations under discrete measures and converge to the continuous operator as the observation grid is refined. Third, we show that estimation under the discrete-measure formulation reduces to standard multivariate regression, with statistical properties governed by classical results. A simulation study illustrates these principles, highlighting the roles of discretization, conditioning, and estimation. Overall, the proposed framework clarifies the relationship between functional and multivariate regression and provides a meaningful interpretation of discretized  modeling approaches as operator estimation under different measure specifications. This perspective also explains why vectorized multivariate regression is often competitive with functional methods in linear settings: it directly estimates the discrete-measure representation of the underlying operator.
\end{abstract}

\vspace{0.5em}

\section{Introduction}

Functional data analysis (FDA) provides a framework for modeling relationships between functional predictors and responses. Classical treatments organize regression models into distinct categories, including scalar-on-function, function-on-scalar, function-on-function, and multivariate regression, each with its own estimation methods and theoretical guarantees (Ramsay and Silverman, 2005; Hsing and Eubank, 2015; Kokoszka and Reimherr, 2017). While this taxonomy has been useful for organizing the literature, it reflects differences in representation rather than fundamental differences in model structure.

Although these models are often presented separately, their mathematical forms already suggest a common underlying structure: each relates predictors and responses through linear functionals or operators acting on functions.

A prototypical example is scalar-on-function regression, in which a scalar response is related to a functional predictor through
\begin{equation}
Y
=
\alpha
+
\int_S \beta(s)X(s)\,ds
+
\varepsilon,
\label{eq:sof_intro}
\end{equation}
where \(X(s)\) is a functional predictor and \(\beta(s)\) is a coefficient function describing how different regions of the predictor contribute to the scalar response.

More generally, both the predictor and response may be functional. In its classical linear form, function-on-function regression is written as
\begin{equation}
Y(t)
=
\alpha(t)
+
\int_S \beta(s,t)X(s)\,ds
+
\varepsilon(t),
\label{eq:fof_intro}
\end{equation}
where \(X(s)\) is a functional predictor, \(Y(t)\) is a functional response, and \(\beta(s,t)\) is a coefficient surface describing how values of the predictor at location \(s\) contribute to the response at location \(t\). Scalar-on-function, function-on-scalar, and multivariate regression models can be viewed as lower-dimensional or discretized versions of this same general structure.

The existing literature already contains indications that these models are not fundamentally separate and repeatedly points toward a deeper commonality among them. In practice, functional regression methods are almost always implemented through discretized or finite-dimensional representations, while theoretical developments often continue to treat the resulting models as distinct categories. Discretely observed functional data are typically treated as approximations to underlying continuous functions, with theoretical analysis focusing on discretization error (Hsing and Eubank, 2015). It has also been noted that scalar and multivariate responses can be embedded within a functional framework (Ramsay and Silverman, 2005). More recently, operator-based formulations have been proposed (Zhou, Lai, and Kong, 2023), highlighting a common structural representation for functional regression models. However, these approaches are generally developed under continuous-measure settings and do not explicitly establish equivalence with multivariate regression arising from discretized functional observations.

In this paper, we propose a unified operator framework in which regression is modeled as an integral operator acting between $L^2$ spaces defined with respect to general measures. Under this perspective, the apparent distinctions between scalar, multivariate, and functional regression arise from differences in the underlying measures rather than differences in model structure.

Within this framework, we establish three main results. First, we show that the standard regression taxonomy can be expressed as a single operator under different choices of input and output measures. Second, we demonstrate that discrete weighted computations correspond to exact evaluations of the operator under a discrete measure, rather than approximations of a continuous model, and that these discrete representations converge to the continuous operator as the observation grid is refined. Third, we show that estimation under the discrete-measure formulation reduces to standard multivariate regression, with statistical properties governed by classical results.

This perspective provides a meaningful interpretation of functional and multivariate regression and clarifies the role of discretization. In particular, it shows that differences between commonly used methods arise from estimation strategy, conditioning, and measure specification, rather than from fundamentally different regression models.

Before presenting the unifying result, we introduce the notation needed to describe regression models under different choices of measure.

\section{Setup and Notation}

 The purpose of this setup is not to develop measure theory in full generality, but to provide a common language in which continuous functional models, discretized functional observations, and multivariate regression can be written as instances of the same operator. In this framework, the measure determines whether integration corresponds to a continuous integral, a weighted sum, or a combination of the two.

Let $(S,\mathcal{B}_S,\mu_X)$ and $(T,\mathcal{B}_T,\mu_Y)$ be measurable spaces equipped with $\sigma$-finite measures $\mu_X$ and $\mu_Y$. Here, $S$ and $T$ denote the domains of the predictor and response functions, respectively, while $\mathcal{B}_S$ and $\mathcal{B}_T$ denote collections of measurable subsets on those domains. The measures $\mu_X$ and $\mu_Y$ specify how functions are integrated and compared over $S$ and $T$. The domains \(S\) and \(T\) may be continuous, discrete, or mixed, depending on whether the corresponding variables are treated as functions, vectors, or scalar quantities.

A Hilbert space is a vector space with an inner product that allows functions to be treated geometrically, much like vectors in ordinary Euclidean space.  Define the Hilbert spaces {\small
\[
L^2(\mu_X)
=
\left\{
f : S \to \mathbb{R}
\;\bigg|\;
\int_S [f(s)]^2 \, d\mu_X(s) < \infty
\right\} \mbox{ and }
L^2(\mu_Y)
=
\left\{
g : T \to \mathbb{R}
\;\bigg|\;
\int_T [g(t)]^2 \, d\mu_Y(t) < \infty
\right\}.
\] }Intuitively, the spaces $L^2(\mu_X)$ and $L^2(\mu_Y)$ consist of square-integrable functions under the corresponding measures; that is, functions whose squared values have finite integral. These spaces provide a geometric setting in which functions can be added, projected, and compared using inner products, much like vectors in ordinary Euclidean space. The Hilbert space structure allows regression relationships between functions to be expressed through linear operators and projections. For general background on Hilbert spaces and functional analysis, see Kreyszig (1978) or Conway (1990). These spaces are equipped with the inner products
\begin{equation}
\langle f_1, f_2 \rangle_{L^2(\mu_X)}
=
\int_S f_1(s) f_2(s)\, d\mu_X(s),
\qquad \text{and} \qquad
\langle g_1, g_2 \rangle_{L^2(\mu_Y)}
=
\int_T g_1(t) g_2(t)\, d\mu_Y(t),
\label{eq:innerproducts}
\end{equation}
which generalize the usual Euclidean inner product to functional settings.

Let $\beta \in L^2(\mu_X \times \mu_Y)$ define the operator
\begin{equation}
(T_\beta X)(t)
=
\int_S \beta(s,t)\, X(s)\, d\mu_X(s),
\label{eq:operator}
\end{equation}
mapping functions in $L^2(\mu_X)$ to functions in $L^2(\mu_Y)$.  For simplicity of notation, we suppress intercept terms in the operator formulation. Intercepts may be incorporated separately, as in the classical function-on-function model \eqref{eq:fof_intro}, without affecting the operator representation developed in this paper.

We consider the regression model
\begin{equation}
Y(t)
=
(T_\beta X)(t) + \varepsilon(t),
\label{eq:model}
\end{equation}
where $\varepsilon(t)$ denotes a residual process.

A discrete measure has the form
\begin{equation}
\mu_X = \sum_{i=1}^{p} w_i \delta_{s_i},
\label{eq:discrete}
\end{equation}
where $\delta_{s_i}$ denotes a Dirac measure concentrated at the point $s_i$ and $w_i > 0$ are weights.

Under this specification,
\begin{equation}
\int_S f(s)\, d\mu_X(s)
=
\sum_{i=1}^{p} w_i f(s_i).
\label{eq:sum}
\end{equation}

Of course, when $\mu_X$ is Lebesgue measure on an interval, the integral reduces to the standard continuous integral
\[
\int_S f(s)\, d\mu_X(s)
=
\int_S f(s)\, ds,
\]
which may be interpreted in the usual Riemann or Lebesgue sense for sufficiently regular functions.

The use of general measures allows continuous, discrete, and mixed representations to be treated within a common notation. When $\mu_X$ is Lebesgue measure, integrals take the usual continuous form; when $\mu_X$ is discrete, integrals reduce to weighted sums. Thus, the measure determines how functions are aggregated over their domains. For a rigorous treatment of integration with respect to general measures, see Cohn (1980) or Rudin (1987).

\section{Unified Operator Framework}

Classical FDA literature typically treats scalar, multivariate, and functional regression as distinct model classes (Ramsay and Silverman, 2005; Hsing and Eubank, 2015). Operator-based formulations have also been proposed; for example, Zhou, Lai, and Kong (2023) formulate functional regression through a linear operator framework. However, these formulations are generally developed in continuous-measure settings and do not explicitly establish equivalence with multivariate regression arising from discretized functional observations.

The central idea of this work is that regression can be formulated as an integral operator acting between $L^2$ spaces defined with respect to general measures. Under this perspective, the apparent distinctions between scalar, multivariate, and functional regression arise from differences in the underlying measures rather than differences in model structure. The following theorem demonstrates that this framework recovers the standard regression taxonomy as special cases.

\begin{theorem}[Operator--Measure Unification of Functional and Multivariate Regression]

Let $(S,\mathcal{B}_S,\mu_X)$ and $(T,\mathcal{B}_T,\mu_Y)$ be measurable spaces equipped with $\sigma$-finite measures $\mu_X$ and $\mu_Y$. Let $\beta \in L^2(\mu_X \times \mu_Y)$ define the operator in \eqref{eq:operator}, and consider the regression model in \eqref{eq:model}.

Then the following regression models arise as special cases obtained by specifying the measures $\mu_X$ and $\mu_Y$:

\begin{enumerate}

\item \textbf{Function-to-function regression.}  
If $\mu_X = \lambda_S$ and $\mu_Y = \lambda_T$, then
\[
Y(t) = \int_S \beta(s,t)\, X(s)\, ds + \varepsilon(t).
\]

\item \textbf{Scalar-on-function regression.}  
If $\mu_Y = \delta_{t_0}$ for some fixed $t_0 \in T$, then
\[
Y = \int_S \beta(s,t_0)\, X(s)\, ds + \varepsilon.
\]

\item \textbf{Function-on-vector regression.}
If $
\mu_X=\sum_{i=1}^{p}\delta_{s_i},$ and
$\mu_Y=\lambda_T,
$ 
then
\[
Y(t)
=
\sum_{i=1}^{p}\beta(s_i,t)X(s_i)+\varepsilon(t),
\]
which corresponds to regression of a functional response on a finite-dimensional vector of predictors. The classical function-on-scalar model arises as the special case \(p=1\).

\item \textbf{Multivariate multiple regression (discretized functional regression).}  
If
\[
\mu_X = \sum_{i=1}^{p} w_i \delta_{s_i}, \qquad
\mu_Y = \sum_{j=1}^{q} v_j \delta_{t_j},
\]
then for each $t_j$,
\[
Y(t_j) = \sum_{i=1}^{p} w_i \beta(s_i,t_j)\, X(s_i) + \varepsilon(t_j),
\quad j = 1,\dots,q,
\]
which yields the matrix form
\[
\mathbf{Y} = \mathbf{X} B + \mathbf{E}.
\]

\item \textbf{Scalar-to-scalar regression.}  
If $\mu_X = \delta_{s_0}$ and $\mu_Y = \delta_{t_0}$, then
\[
Y = \beta(s_0,t_0)\, X + \varepsilon.
\]

\end{enumerate}

\end{theorem}

\begin{proof}
We verify each case explicitly by substituting the corresponding measure into \eqref{eq:operator}.

\textbf{Function-to-function.} If $\mu_X = \lambda_S$ and $\mu_Y = \lambda_T$, then \eqref{eq:operator} reduces to the standard Lebesgue integral, yielding
\[
(T_\beta X)(t) = \int_S \beta(s,t) X(s)\, ds,
\]
which is the classical function-on-function regression model.

\textbf{Scalar-on-function.} If $\mu_Y = \delta_{t_0}$, evaluating at $t_0$ gives
\[
Y = \int_S \beta(s,t_0) X(s)\, ds + \varepsilon,
\]
which is the scalar-on-function model.

\textbf{Function-on-vector.} If $\mu_X = \sum_i \delta_{s_i}$, then using \eqref{eq:sum},
\[
(T_\beta X)(t) = \sum_i \beta(s_i,t) X(s_i),
\]
which corresponds to regression of a functional response on a finite-dimensional vector of predictors.

\textbf{Multivariate regression.} If both $\mu_X$ and $\mu_Y$ are discrete, then applying \eqref{eq:sum} yields
\[
Y(t_j) = \sum_i w_i \beta(s_i,t_j) X(s_i),
\]
which is exactly the multivariate multiple regression model.

\textbf{Scalar regression.} If both measures are Dirac measures, the result follows immediately.

\end{proof}

\begin{corollary}[Common Estimand]
An immediate implication of Theorem 1 is that classical FDA methods, multivariate regression, and modern machine learning approaches can be viewed as estimators of the same underlying operator $\beta$. Differences in empirical performance should therefore be interpreted as differences in estimation strategy, rather than differences in model specification.
\end{corollary}

\begin{remark}[Scope: vectorized functional observations]
Theorem 1 does not assert that all multivariate regression problems arise from functional data. Rather, it formalizes the common practice of vectorizing functional observations: when predictor and response functions are observed on finite grids, their sampled values may be treated as vectors and analyzed using multivariate regression. The theorem shows that such vectorized regression models are mathematically equivalent to functional regression models evaluated under discrete measures.
\end{remark}

\begin{remark}[Role of the measure]
The apparent differences between regression formulations arise from the choice of measure defining the domain. Continuous measures yield functional representations, while discrete measures yield finite-dimensional vector representations. Thus, the structure of the model is unchanged; only the underlying measure differs.
\end{remark}

\begin{remark}[Discretization as a measure]
For finite observation grids, discretization defines a measure rather than an approximation. If
\[
\mu_X = \sum_i w_i \delta_{s_i},
\]
then
\[
\int_S f(s)\, d\mu_X(s) = \sum_i w_i f(s_i).
\]
Thus, weighted sums arising from discretized functional data are exact integrals with respect to a discrete measure. They are not numerical approximations of Lebesgue integrals unless one explicitly considers a limiting regime in which the grid is refined.
\end{remark}

\medskip

Theorem 1 establishes the structural equivalence of regression models under different measures. We next examine how this operator is realized computationally under discretized observations.

\subsection{Scalar-on-Function Regression under Discrete Measures}

We begin with the scalar-on-function regression model introduced in \eqref{eq:sof_intro}. Under the operator framework developed above, this model may be written as
\[
Y
=
\int_S \beta(s)X(s)\,d\mu_X(s)
+
\varepsilon,
\]
where $\mu_X$ specifies the measure under which the predictor function is integrated. In practice, scalar-on-function regression is rarely implemented through direct regression on densely sampled functional observations due to strong correlation and ill-conditioning among neighboring grid points. Instead, FDA methods typically represent either the predictor function \(X(s)\), the coefficient function \(\beta(s)\), or both through finite-dimensional basis expansions such as splines, Fourier bases, wavelets, or functional principal components. Estimation then proceeds through regression on the resulting basis coefficients, often with additional smoothing or regularization. Under the present framework, these approaches may be viewed as alternative parameterizations and estimation strategies for the same underlying operator evaluated under different representations of the measure space.

It is important to distinguish between regression on discretized functional observations and regression on basis scores or functional principal component coefficients. Direct discretization regresses on the sampled values
\[
X(s_1),\dots,X(s_p),
\]
whereas basis-based FDA methods first represent the functional observations, the coefficient function, or both in a lower-dimensional basis. For example, one may write
\[
X_i(s)\approx \sum_{k=1}^K z_{ik}\phi_k(s),
\]
where \(\{\phi_k\}\) denotes a chosen functional basis and
\[
z_{ik}
=
\int_S X_i(s)\phi_k(s)\,ds
\]
is the coefficient or score of the \(i\)th function along the \(k\)th basis direction. Regression is then performed on the finite-dimensional scores \(z_{ik}\), rather than directly on the sampled grid values. Although estimation may proceed through basis coefficients or scores, the resulting fitted values may still be represented through weighted combinations of the original functional observations via the implied coefficient function. These approaches therefore correspond to different finite-dimensional representations of the same underlying operator.

Suppose now that the predictor function is observed on a finite grid
\[
s_1,\ldots,s_p \in S,
\]
and define the discrete measure
\[
\mu_X
=
\sum_{i=1}^{p} w_i \delta_{s_i},
\]
with weights $w_i > 0$.

Under this specification,
\[
Y
=
\int_S \beta(s)X(s)\,d\mu_X(s)
+
\varepsilon
=
\sum_{i=1}^{p} w_i \beta(s_i)X(s_i)
+
\varepsilon.
\]

Define the vectorized predictor
\[
\mathbf{x}
=
\begin{pmatrix}
X(s_1)\\
\vdots\\
X(s_p)
\end{pmatrix},
\]
and define the coefficient vector
\[
b_i = w_i \beta(s_i).
\]

Then the model becomes
\[
Y = \mathbf{x}^\top b + \varepsilon,
\]
which is precisely the standard linear regression model.

\medskip

Thus, scalar-on-function regression evaluated under a discrete measure reduces exactly to ordinary linear regression on the vectorized functional observations. In this representation, the coefficient vector $b$ is the discrete-measure representation of the coefficient function $\beta(s)$.

\medskip

When \(X^\top X\) is nonsingular, the ordinary least squares estimator of the discrete coefficient vector is
\[
\hat{b}
=
(X^\top X)^{-1}X^\top Y,
\qquad
\hat{b}\in\mathbb{R}^{p}.
\]

This estimator is the standard least-squares estimator applied to the vectorized functional observations. Under the discrete-measure representation, estimation therefore proceeds through familiar linear regression machinery. As in ordinary regression, correlation among the discretized predictor values primarily affects conditioning and multicollinearity, while standard least-squares properties hold under the usual assumptions on the residual term. Formal consistency under the discrete-measure representation is discussed in Section~5.

\medskip

This simple case already illustrates the central idea of the framework: once the predictor domain is equipped with a discrete measure, the resulting weighted sums are exact integrals under that measure rather than approximations to a continuous integral.

\subsection{From Operator to Multivariate Regression}

To make the connection between the operator formulation and standard multivariate regression explicit, consider the case where both the predictor and response functions are observed on finite grids. Let
\[
s_1,\ldots,s_p \in S, \qquad t_1,\ldots,t_q \in T,
\]
and define the corresponding discrete measures
\[
\mu_X = \sum_{i=1}^{p} w_i \delta_{s_i}, \qquad
\mu_Y = \sum_{j=1}^{q} v_j \delta_{t_j},
\]
with $w_i, v_j > 0$.

Under this specification, the operator model \eqref{eq:model} evaluated at the grid points $t_j$ becomes
\[
Y(t_j)
=
\int_S \beta(s,t_j)\, X(s)\, d\mu_X(s) + \varepsilon(t_j)
=
\sum_{i=1}^{p} w_i \beta(s_i,t_j)\, X(s_i) + \varepsilon(t_j),
\quad j = 1,\ldots,q.
\]

Define the vectorized observations
\[
\mathbf{x}
=
\begin{pmatrix}
X(s_1)\\
\vdots\\
X(s_p)
\end{pmatrix},
\qquad
\mathbf{y}
=
\begin{pmatrix}
Y(t_1)\\
\vdots\\
Y(t_q)
\end{pmatrix},
\qquad
\boldsymbol{\varepsilon}
=
\begin{pmatrix}
\varepsilon(t_1)\\
\vdots\\
\varepsilon(t_q)
\end{pmatrix}.
\]

Define the coefficient matrix
\[
B_{ij} = w_i \beta(s_i,t_j).
\]
Then the model can be written as
\[
\mathbf{y} = B^\top \mathbf{x} + \boldsymbol{\varepsilon}.
\]

For $n$ independent observations, stacking the data yields the standard multivariate multiple regression model
\[
Y = X B + E,
\]
where $X \in \mathbb{R}^{n \times p}$ and $Y \in \mathbb{R}^{n \times q}$.

\medskip

The weights $w_i$ are part of the measure specification and are not merely numerical integration artifacts. They represent the modeler's choice of how observations on the grid are weighted, and thus define the measure under which the operator is evaluated.

We incorporate the weights $w_i$ into the definition of $B$ for notational simplicity. Alternatively, one could define $B_{ij} = \beta(s_i,t_j)$ and treat the weights separately; the two representations are equivalent but have different implications for estimation and identifiability.

The output weights $v_j$ do not affect pointwise evaluation at the grid points, but become relevant when the response is integrated against test functions, for example in prediction error measured in $L^2(\mu_Y)$.

\medskip

\noindent
\textbf{Key interpretation.} The matrix $B$ is not a separate object from the functional coefficient surface $\beta(s,t)$; it is the discrete-measure representation of that surface evaluated on the input and output grids.

\medskip

This representation makes precise the sense in which multivariate regression arising from discretized (vectorized) functional observations and functional regression estimate the same underlying operator under different measure specifications.

When \(X^\top X\) is nonsingular, the ordinary least squares estimator of the discrete coefficient matrix is
\[
\hat{B}
=
(X^\top X)^{-1}X^\top Y,
\qquad
\hat{B}\in\mathbb{R}^{p\times q}.
\]

This estimator is the standard multivariate least-squares estimator applied to the vectorized functional observations. Under the discrete-measure representation, estimation therefore proceeds through familiar multivariate regression machinery, although statistical efficiency and inference depend on the covariance structure of the residual process \(E\). Formal consistency under this representation is discussed in Section~5.
\subsection{Connection to Neural Functional Models}

The discrete-measure formulation developed above also provides a useful perspective on
recent neural network approaches to functional regression.

Many neural functional models, including coefficient-based functional neural networks
(Thind, Multani, and Cao, 2023) and adaptive functional neural networks
(Yao, Mueller, and Wang, 2021), operate on discretized observations of functional inputs.
These models typically construct intermediate representations through inner products of the form
\[
\tilde{c}_k = \int_S \beta_k(s)\, X(s)\, ds,
\]
which, under discretization, are computed as
\[
\tilde{c}_k = \sum_{i=1}^{p} \beta_k(s_i)\, X(s_i).
\]
Here $\{s_i\}$ denotes the input grid.

In this formulation, the functions $\beta_k(s)$ represent learned functional weights (e.g., basis functions or neural network outputs), rather than slices of a coefficient surface indexed by an output grid.

Comparing this with the discrete-measure representation of the operator in Section 3.1,
\[
Y(t_j) = \sum_{i=1}^{p} w_i \beta(s_i,t_j)\, X(s_i),
\]
we see that neural functional models compute inner products that correspond to evaluations of integrals under a discrete input measure. Thus, the summations that arise in these architectures are exact evaluations of integrals under the discrete measure $\mu_X$.

This perspective suggests that neural functional models are not fundamentally different
from classical regression formulations, but rather represent alternative estimation
strategies for the same underlying operator. The distinction lies in how the coefficient
surface $\beta$ is parameterized and learned, rather than in the structure of the model
itself.

For architectures with explicit output-grid structure, such as neural operator models, a full coefficient surface $\beta(s,t)$ is represented explicitly, and the correspondence with the operator formulation in Theorem 1 becomes exact.

\section{Operator Representation under General Measures}

The preceding section established that regression models across functional, multivariate, and scalar settings arise from a common operator under different measure specifications. We now formalize how this operator is represented under discrete measures and how these representations relate to the continuous operator as the observation grid is refined.

\begin{definition}
A quadrature rule $\{s_i^{(m)}, w_i^{(m)}\}_{m \ge 1}$ on $(S,\lambda_S)$
is said to be \emph{consistent on a class} $\mathcal{F}$ of $\lambda_S$-integrable
functions if, for every $f \in \mathcal{F}$,
\[
\sum_{i=1}^{p_m} w_i^{(m)} f(s_i^{(m)})
\longrightarrow
\int_S f(s)\,d\lambda_S(s).
\]
\end{definition}

\begin{theorem}[Exact Discrete Representation and Continuous Convergence]

Let $S \subset \mathbb{R}$ be compact, and let $\lambda_S$ denote Lebesgue measure on $S$. Let $(T,\mathcal{B}_T,\mu_Y)$ be a measurable space equipped with a $\sigma$-finite measure $\mu_Y$. Let
\[
\mu_X^{(m)}
=
\sum_{i=1}^{p_m} w_i^{(m)} \delta_{s_i^{(m)}}
\]
be a sequence of discrete measures on $S$. Define the continuous operator
\[
(T_\beta X)(t)
=
\int_S \beta(s,t)X(s)\,d\lambda_S(s),
\]
and define the corresponding discrete-measure operators
\[
(T_\beta^{(m)}X)(t)
=
\int_S \beta(s,t)X(s)\,d\mu_X^{(m)}(s).
\]

Then:

\begin{enumerate}

\item \textbf{Exact discrete representation.}

For each fixed $m$,
\[
(T_\beta^{(m)}X)(t)
=
\sum_{i=1}^{p_m}
w_i^{(m)}
\beta(s_i^{(m)},t)
X(s_i^{(m)}).
\]

Thus, the weighted sums used in practice with discretized functional observations are exact integrals against $\mu_X^{(m)}$, with no approximation error at the discrete level.

\item \textbf{Convergence to the continuous operator.}

Suppose that, for $\mu_Y$-almost every $t\in T$, the function
\[
h_t(s)=\beta(s,t)X(s)
\]
belongs to a class $\mathcal{F}$ on which the quadrature rule
$\{s_i^{(m)},w_i^{(m)}\}$ is consistent; that is,
\[
\sum_{i=1}^{p_m}
w_i^{(m)}h_t(s_i^{(m)})
\longrightarrow
\int_S h_t(s)\,d\lambda_S(s).
\]

Assume further that there exists a measurable function
\[
H:S\times T\to [0,\infty)
\]
such that
\[
|\beta(s,t)X(s)|
\le
H(s,t)
\]
for all $s\in S$ and for $\mu_Y$-almost every $t\in T$, with
\[
\int_T
\left(
\int_S H(s,t)\,d\lambda_S(s)
\right)^2
d\mu_Y(t)
<
\infty.
\]

Finally, assume that the quadrature rule applied to $H(\cdot,t)$ is uniformly bounded in the sense that there exists
\[
M(t)\in L^2(\mu_Y)
\]
such that
\[
\int_S H(s,t)\,d\mu_X^{(m)}(s)
\le
M(t)
\]
for all sufficiently large $m$ and for $\mu_Y$-almost every $t$.

Then
\[
T_\beta^{(m)}X
\longrightarrow
T_\beta X
\quad \text{in } L^2(\mu_Y).
\]

\end{enumerate}

\end{theorem}

\begin{proof}

\textit{Part (i): Exact discrete representation.}

By the definition of integration with respect to a discrete measure, for any measurable function $f$ on $S$,
\[
\int_S f(s)\,d\mu_X^{(m)}(s)
=
\sum_{i=1}^{p_m}
w_i^{(m)} f(s_i^{(m)}).
\]
Taking $f(s)=\beta(s,t)X(s)$ for fixed $t\in T$ yields
\[
(T_\beta^{(m)}X)(t)
=
\sum_{i=1}^{p_m}
w_i^{(m)}
\beta(s_i^{(m)},t)
X(s_i^{(m)}),
\]
which establishes the exact discrete representation.

\medskip
\noindent\textit{Part (ii): Convergence to the continuous operator.}

Fix $t\in T$ in the full-measure set on which the quadrature consistency assumption holds for
\[
h_t(s)=\beta(s,t)X(s).
\]
By assumption,
\[
\sum_{i=1}^{p_m}
w_i^{(m)} h_t(s_i^{(m)})
\longrightarrow
\int_S h_t(s)\,d\lambda_S(s),
\]
which is precisely
\[
(T_\beta^{(m)}X)(t)
\longrightarrow
(T_\beta X)(t)
\qquad
\text{for $\mu_Y$-almost every } t\in T.
\tag{$\ast$}
\]

\medskip
\noindent\textit{Domination.}
The pointwise envelope $|\beta(s,t)X(s)| \le H(s,t)$ gives
\[
|(T_\beta X)(t)|
\le
\int_S H(s,t)\,d\lambda_S(s)
=:G(t),
\]
and the uniform quadrature bound gives, for all $m$ sufficiently large,
\[
|(T_\beta^{(m)}X)(t)|
\le
\int_S H(s,t)\,d\mu_X^{(m)}(s)
\le
M(t).
\]
Hence
\[
|(T_\beta^{(m)}X)(t)-(T_\beta X)(t)|
\le
M(t)+G(t),
\]
and therefore
\[
|(T_\beta^{(m)}X)(t)-(T_\beta X)(t)|^2
\le
2M(t)^2+2G(t)^2.
\]
By hypothesis, $M\in L^2(\mu_Y)$, so $M^2\in L^1(\mu_Y)$. Also,
\[
\int_T G(t)^2\,d\mu_Y(t)
=
\int_T
\left(
\int_S H(s,t)\,d\lambda_S(s)
\right)^2
d\mu_Y(t)
<
\infty.
\]
Thus, the squared difference is dominated by an $L^1(\mu_Y)$ function.

\medskip
\noindent\textit{Conclusion.}
Combining the pointwise convergence $(\ast)$ with the domination above, the dominated convergence theorem gives
\[
\int_T
\left|
(T_\beta^{(m)}X)(t)
-
(T_\beta X)(t)
\right|^2
d\mu_Y(t)
\longrightarrow 0.
\]
Therefore,
\[
T_\beta^{(m)}X
\longrightarrow
T_\beta X
\quad \text{in } L^2(\mu_Y).
\]

\end{proof}

\begin{corollary}[Connection to Multivariate Regression]
Combining Theorem 2 with the construction in Section 3.1 shows that the multivariate regression model
\[
Y = X B + E
\]
is precisely the discrete-measure representation of the operator $T_\beta$. In this representation, the matrix $B$ is the discrete-measure representation of the coefficient surface $\beta(s,t)$, with the input weights incorporated into its definition.
\end{corollary}

\begin{remark}[Discrete computation as exact evaluation]
Part (i) shows that weighted sums arising from discretized functional data are not approximations of continuous-measure integrals, but exact evaluations of the operator under a discrete measure. This contrasts with the standard view in functional data analysis, where discretization is typically treated as an approximation to an underlying continuous model (Hsing and Eubank, 2015).
\end{remark}

\begin{remark}[Discrete-to-continuous transition]
Part (ii) shows that discrete representations form a sequence of operators converging to the continuous operator as the grid is refined. Thus, approximation enters only when comparing operators defined under different measures, not within a fixed discrete-measure formulation.
\end{remark}

\begin{remark}[Connection to Neural Functional Models]
Theorem 2 also provides an operator-theoretic interpretation of neural functional models operating on discretized inputs. As discussed in Section 3.3, these architectures compute weighted sums over sampled functional observations and therefore act on the same underlying discrete-measure operator, differing primarily in parameterization and estimation strategy.
\end{remark}

\medskip

Theorem 2 separates two ideas that are often conflated. For any fixed grid, the weighted sums used in computation are exact evaluations of the operator under a discrete measure. Approximation arises only when comparing this discrete-measure operator to a continuous-measure operator as the grid is refined.
\section{Estimation under Discrete Measures}

We now consider estimation of the operator under the discrete-measure representation
introduced in Sections 3 and 4. In this setting, the functional regression model reduces
to a multivariate multiple regression model.

\begin{theorem}[Consistency under Discrete-Measure Representation]

Suppose that the predictor and response functions are observed on finite grids
$\{s_i\}_{i=1}^p$ and $\{t_j\}_{j=1}^q$, with corresponding discrete measures
\[
\mu_X = \sum_{i=1}^{p} w_i \delta_{s_i}, \qquad
\mu_Y = \sum_{j=1}^{q} v_j \delta_{t_j}.
\]

Let $\{(X_k, Y_k)\}_{k=1}^n$ be independent observations satisfying
\[
Y_k = X_k B + \varepsilon_k,
\]
where $B_{ij} = w_i \beta(s_i, t_j)$ and $\mathbb{E}[\varepsilon_k \mid X_k] = 0$.

Assume that $\mathbb{E}[X_k^\top X_k]$ is nonsingular and
$\mathbb{E}[\|\varepsilon_k\|^2] < \infty$.

Let
\[
\hat{B} = (X^\top X)^{-1} X^\top Y
\]
denote the least squares estimator.

Then
\[
\hat{B} \xrightarrow{P} B
\quad \text{as } n \to \infty.
\]

\end{theorem}

\begin{proof}
Write the least squares estimator as
\[
\hat{B}
=
(X^\top X)^{-1}X^\top Y.
\]
Using the model \(Y = XB + E\), we have
\[
\hat{B}
=
(X^\top X)^{-1}X^\top (XB+E)
=
B + (X^\top X)^{-1}X^\top E.
\]
Equivalently,
\[
\hat{B} - B
=
\left(\frac{1}{n}X^\top X\right)^{-1}
\left(\frac{1}{n}X^\top E\right).
\]
By the law of large numbers,
\[
\frac{1}{n}X^\top X \xrightarrow{p} \mathbb{E}[X_k^\top X_k],
\]
which is nonsingular by assumption. Also, since
\(\mathbb{E}[\varepsilon_k\mid X_k]=0\) and
\(\mathbb{E}\|\varepsilon_k\|^2<\infty\),
\[
\frac{1}{n}X^\top E \xrightarrow{P} 0.
\]
Therefore,
\[
\hat{B} - B \xrightarrow{P} 0,
\]
and hence
\[
\hat{B}\xrightarrow{P}B.
\]
\end{proof}

\begin{remark}[Interpretation]
Theorem 3 shows that estimation under the discrete-measure formulation reduces to
standard multivariate regression. In particular, the estimator $\hat{B}$ consistently
recovers the discrete-measure representation of the operator $\beta$.
\end{remark}

\begin{remark}[Connection to the Framework]
Combining Theorems 1--3, we see that regression models across functional, multivariate,
and neural settings correspond to estimation of the same operator under different measures,
and that standard estimation procedures behave consistently under this representation.
\end{remark}

\section{Simulation Study}

We illustrate the implications of the unified operator framework through a set of controlled simulation studies. The purpose of this section is not to compare methods for predictive performance, but to demonstrate how estimation behaves under the discrete-measure representation developed in Sections 3--5.

All results are interpreted in operator-theoretic terms. In particular, we distinguish between discretization error arising from replacing the continuous measure with a discrete measure, and estimation error arising from finite samples and conditioning of the design.

\subsection{Simulation Setup}

Functional predictors \(X(s)\) are generated on \(S=[0,1]\) using a truncated Karhunen--Lo\`eve expansion. The response is generated according to the operator model
\[
Y(t)
=
\int_0^1 \beta(s,t)\,X(s)\,ds
+
\varepsilon(t),
\]
with \(\varepsilon(t)\) Gaussian noise calibrated to a fixed signal-to-noise ratio.

The functions are observed on grids \(\{s_i\}_{i=1}^p\) and \(\{t_j\}_{j=1}^p\), corresponding to the discrete measure
\[
\mu
=
\frac{1}{p}\sum_{i=1}^p \delta_{s_i}.
\]

Vectorized multivariate regression is then used to estimate the discrete-measure operator \(T_\beta^\mu\).

We report three quantities:
\begin{itemize}
\item discretization error (DiscErr): the difference between \(T_\beta^\mu\) and a high-resolution reference operator,
\item operator error (OpErr): \(\|T_{\hat{\beta}}X - T_\beta X\|\) evaluated on test functions,
\item coefficient error (ISE): the squared \(L^2(\mu\times\mu)\) distance between \(\hat{\beta}\) and \(\beta\).
\end{itemize}
For each experimental configuration, the simulation was repeated over 100 independent Monte Carlo replications, and the reported quantities correspond to averages across these replications. This averaging reduces sensitivity to individual realizations and allows the reported trends to reflect systematic behavior of the estimators under varying discretization, sample size, and conditioning regimes.
\subsection{Discretization and Operator Convergence}

We first examine the effect of grid density. As shown in Figure~\ref{fig:disc_error}, the discretization error decreases at the expected rate as \(p\) increases, closely following a \(p^{-2}\) scaling consistent with midpoint quadrature for smooth integrands. A log--log fit gives an empirical slope of approximately \(-2.0\).

In contrast, the operator error remains approximately constant once the grid becomes sufficiently dense. This behavior reflects the separation between discretization error and estimation error: as the grid is refined, the discrete operator \(T_\beta^\mu\) approaches the continuous operator \(T_\beta\), but the finite-sample estimation error becomes dominant.

These results are consistent with Theorem 2. The discrete representation converges to the continuous operator as the measure becomes increasingly fine, but beyond a certain resolution, additional grid points do not improve performance unless the sample size is increased.

\begin{figure}[h]
\centering
\includegraphics[width=0.75\textwidth]{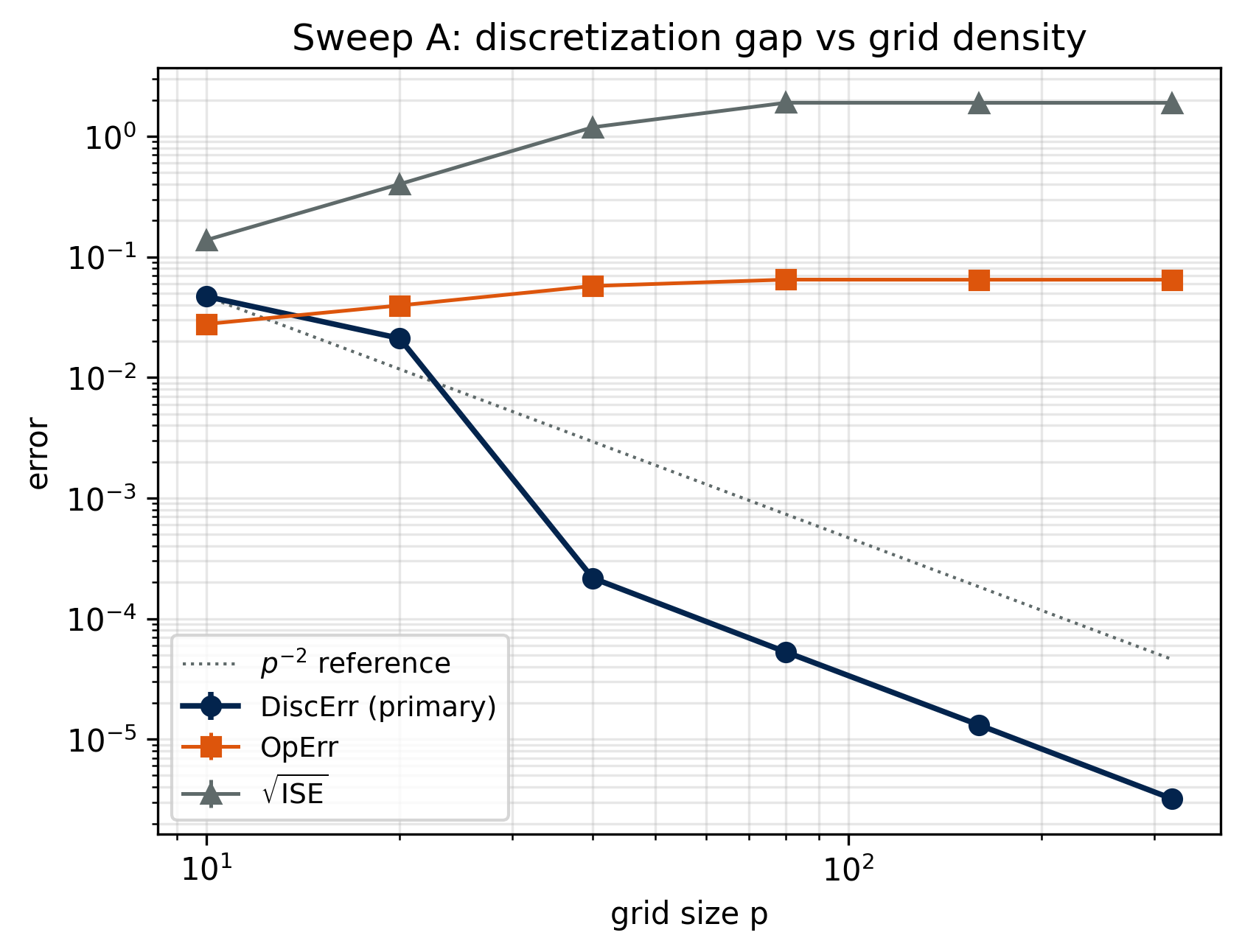}
\caption{Discretization error, operator error, and coefficient error as functions of grid size \(p\). The discretization error decreases at the expected \(p^{-2}\) rate, while operator error stabilizes due to finite-sample estimation error.}
\label{fig:disc_error}
\end{figure}

\begin{table}[h]
\centering
\caption{Summary of discretization and operator error as a function of grid size \(p\).}
\begin{tabular}{ccccc}
\hline
\(p\) & DiscErr & OpErr & RMSE & \(\kappa(X^\top X)\) \\
\hline
10  & 0.0469 & 0.0277 & 0.1956 & 5.3e1 \\
20  & 0.0211 & 0.0395 & 0.1975 & 2.5e2 \\
40  & 0.00022 & 0.0573 & 0.2017 & 1.2e3 \\
80  & 5.3e-05 & 0.0648 & 0.2041 & 9.5e18 \\
160 & 1.3e-05 & 0.0646 & 0.2041 & 2.4e19 \\
320 & 3.2e-06 & 0.0647 & 0.2041 & 1.4e20 \\
\hline
\end{tabular}
\label{tab:disc_summary}
\end{table}

\subsection{Estimation under Increasing Sample Size}

We next examine the effect of sample size at a fixed grid. As shown in Figure~\ref{fig:sample_size}, the operator error decreases at approximately the \(n^{-1/2}\) rate predicted by Theorem 3. A log--log fit gives an empirical slope close to \(-0.5\). In contrast, the prediction error (RMSE) decreases initially but levels off at the noise floor.

This behavior is consistent with the decomposition
\[
\mathbb{E}[\mathrm{RMSE}^2]
\approx
\mathrm{OpErr}^2 + \sigma^2,
\]
which is verified numerically to high precision. The results confirm that operator error captures the estimation component, while prediction error reflects both estimation and irreducible noise.

\begin{figure}[h]
\centering
\includegraphics[width=0.75\textwidth]{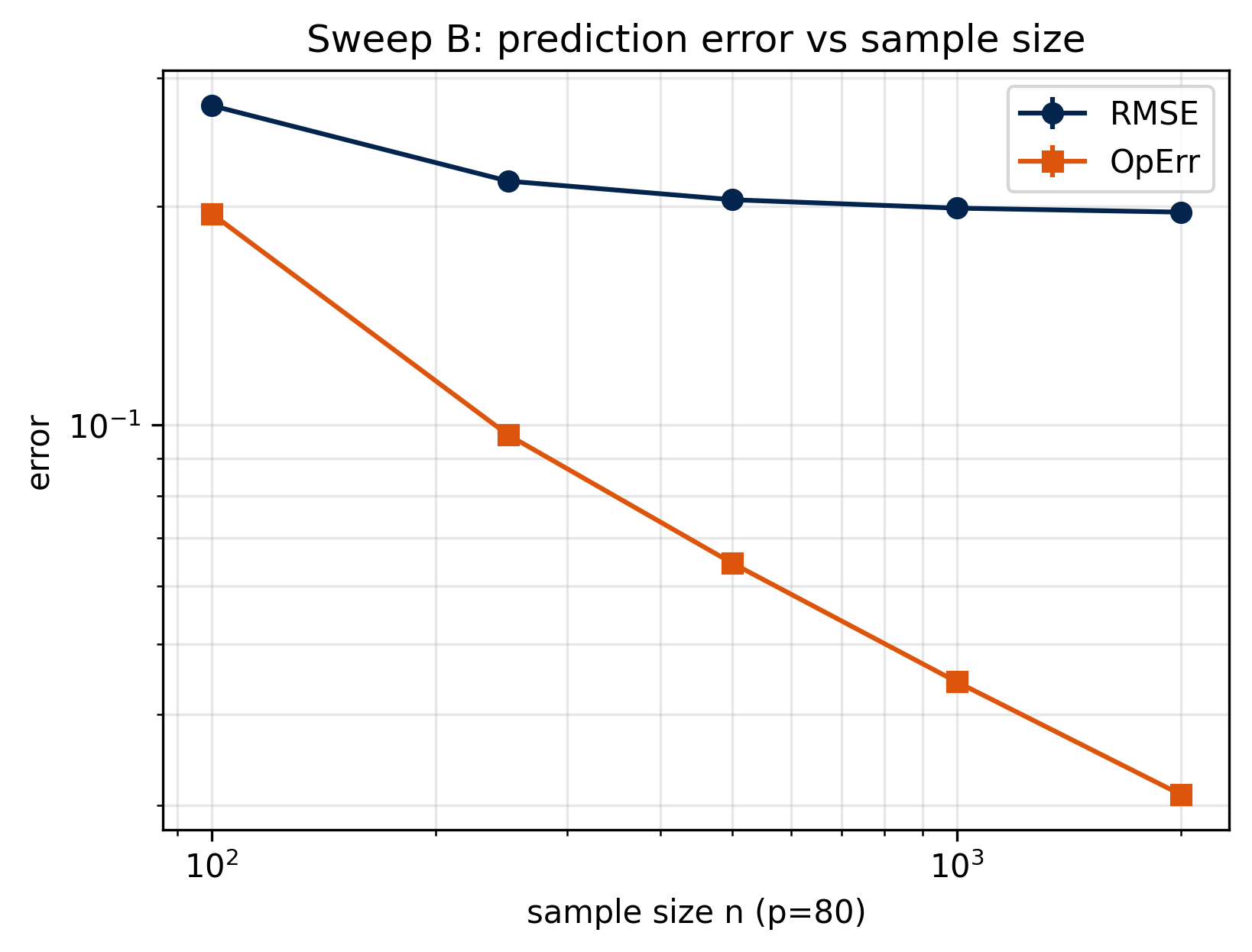}
\caption{Prediction error (RMSE) and operator error as functions of sample size \(n\) at fixed grid size. Operator error decreases at approximately the \(n^{-1/2}\) rate, while prediction error approaches the noise floor.}
\label{fig:sample_size}
\end{figure}

\subsection{Effect of Conditioning and Regularization}

Finally, we examine the role of conditioning. As the grid size \(p\) increases relative to the sample size \(n\), the condition number of \(X^\top X\) grows rapidly due to the strong correlation among neighboring functional observations induced by the smooth Karhunen--Lo\`eve representation. At the same time, the discrete operator more closely approximates the underlying continuous operator as the grid is refined.

These competing effects produce the familiar regression phenomenon in which coefficient estimation becomes increasingly unstable under severe multicollinearity, even though predictive or operator-level performance may remain comparatively accurate. Thus, poor conditioning does not necessarily imply large operator error.

As shown in Figure~\ref{fig:conditioning}, ridge regression provides a more stable estimator by regularizing the empirical Gram operator while still targeting the same underlying operator \(T_\beta^\mu\). Notably, at the largest condition numbers, the OLS operator error can return to levels comparable to ridge regression. This reflects implicit regularization arising from numerical rank truncation in the pseudo-inverse computation: extremely small singular directions are effectively discarded under floating-point precision. In the present simulations, these directions correspond primarily to high-frequency low-variance components induced by the smooth Karhunen--Lo\`eve structure of the predictor process. Ridge regularization and numerical rank truncation therefore represent different mechanisms that can produce similar operator-level behavior while still targeting the same underlying operator.

\begin{figure}[h]
\centering
\includegraphics[width=0.75\textwidth]{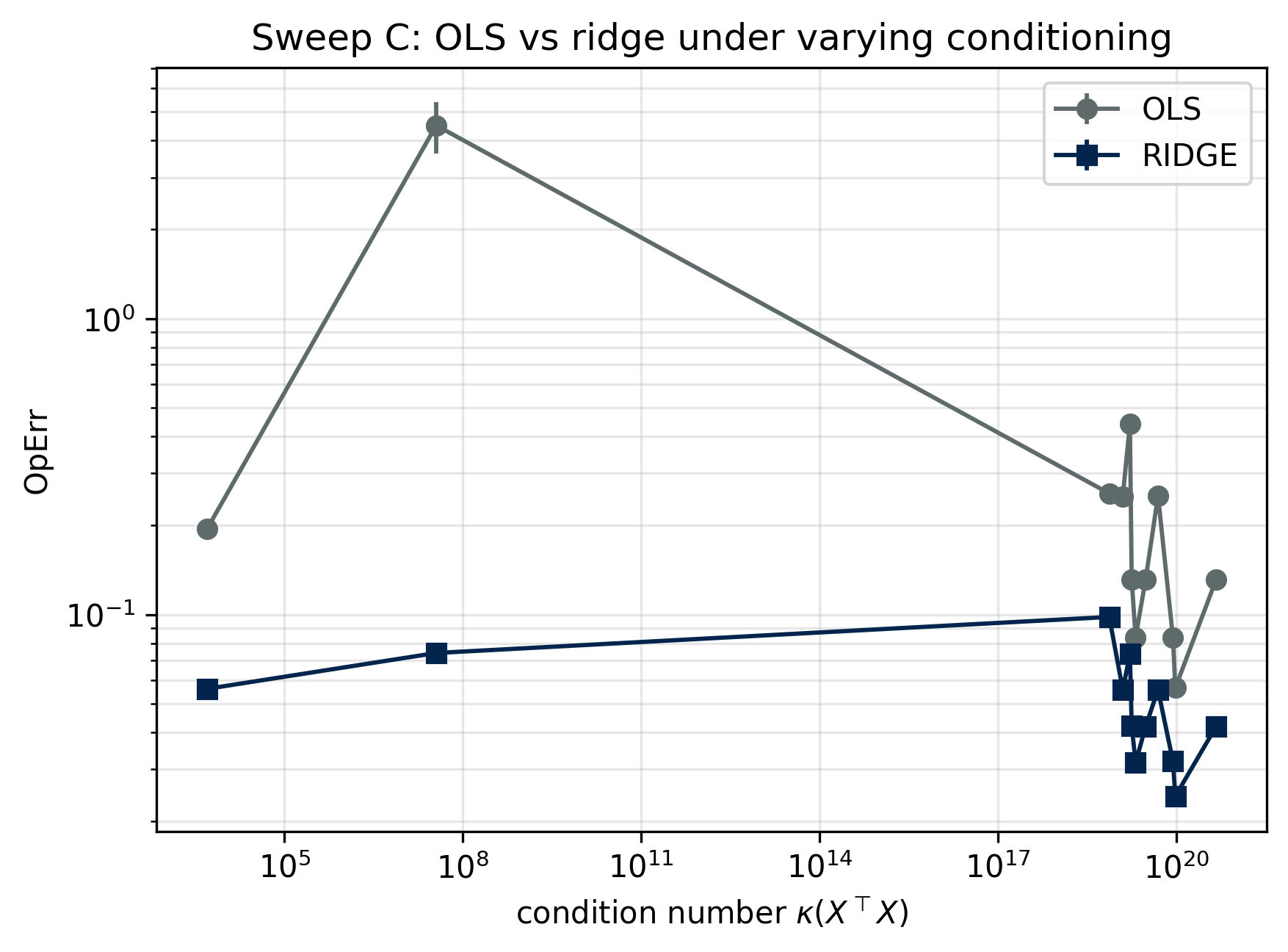}
\caption{Operator error for OLS and ridge regression as a function of the condition number of the design matrix. Although coefficient instability increases with conditioning, operator-level performance can remain comparatively stable.}
\label{fig:conditioning}
\end{figure}

\subsection{Coefficient Error versus Operator Error}

An important observation is that coefficient error (ISE) and operator error behave differently. In regimes with poor conditioning, the estimated coefficient matrix \(\hat{B}\) may deviate substantially from the true discrete coefficient matrix, resulting in large ISE. However, the corresponding operator error remains relatively stable.

This reflects the fact that perturbations in \(\hat{B}\) may have limited impact on the action of the operator on typical inputs. In particular, accurate recovery of the coefficient surface is not necessary for accurate operator evaluation. This distinction underscores the importance of interpreting results at the operator level rather than solely in terms of coefficient estimation.

\begin{figure}[h]
\centering
\includegraphics[width=0.85\textwidth]{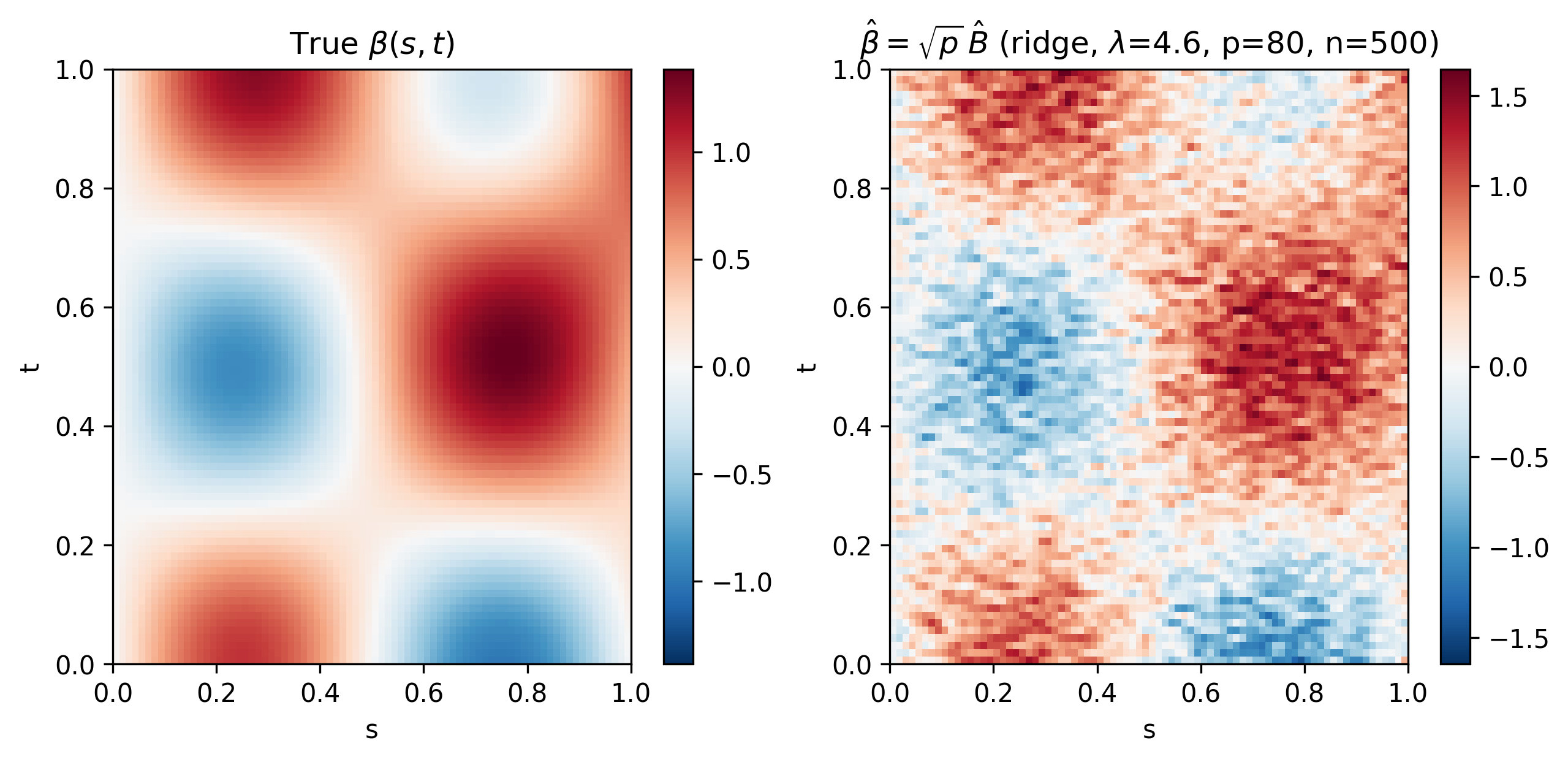}
\caption{True coefficient surface \(\beta(s,t)\) and estimated discrete-measure representation \(\hat{\beta}(s_i,t_j)\). The global structure is recovered despite finite-sample noise.}
\label{fig:beta_heatmap}
\end{figure}

\begin{figure}[h]
\centering
\includegraphics[width=0.9\textwidth]{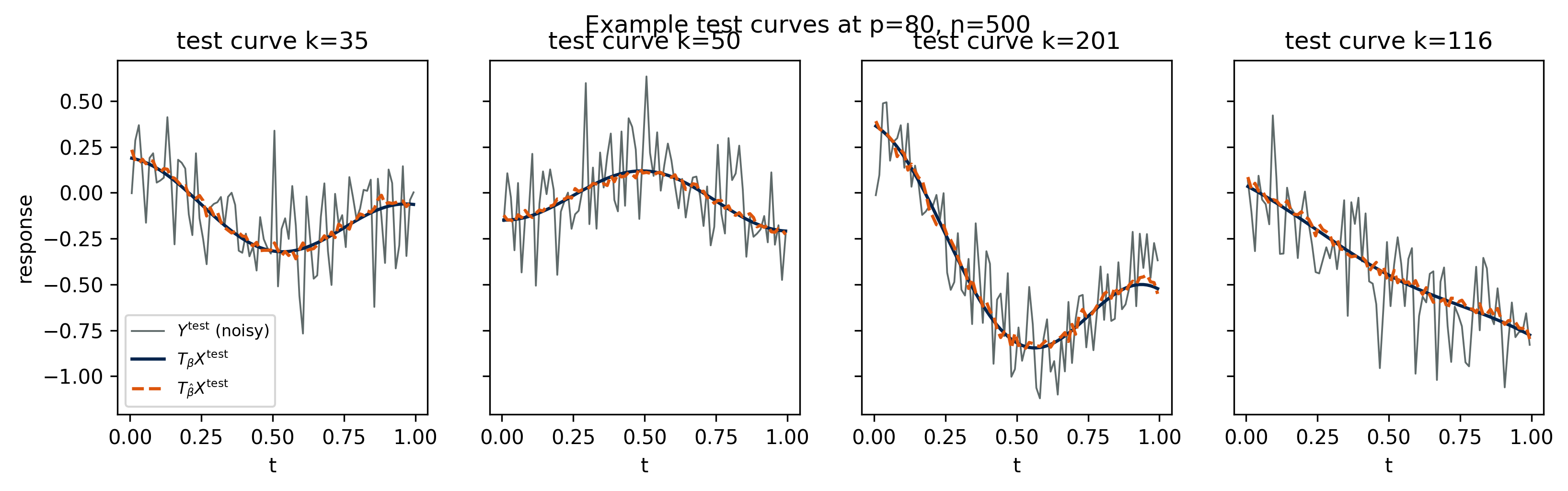}
\caption{Example test curves showing noisy observations, true operator outputs, and predicted responses. The estimated operator captures the underlying structure despite noise.}
\label{fig:example_curves}
\end{figure}

\subsection{Discussion}

The simulation results support the theoretical framework developed in this paper. In particular:

\begin{itemize}
\item Discrete representations correspond to exact evaluations of the operator under a discrete measure, and converge to the continuous operator as the grid is refined.
\item When the grid is sufficiently dense, estimation error rather than discretization error determines performance.
\item Estimation error decreases at the standard \(n^{-1/2}\) rate, consistent with classical multivariate regression theory.
\item Differences between estimators arise from conditioning and regularization, not from differences in the underlying regression model.
\item Accurate operator estimation does not require accurate recovery of the full coefficient surface.
\end{itemize}

Overall, the simulations demonstrate how the measure-theoretic perspective clarifies both the relationship between functional and multivariate regression and the behavior of estimation under discretization.

\section{Discussion}

We have proposed a unified operator framework for functional, multivariate, and scalar regression based on integral operators defined with respect to general measures. Within this framework, classical regression models arise as special cases corresponding to different choices of input and output measures.

The central contribution of this work is to reinterpret discretization not as an approximation to a continuous model, but as a specification of the measure under which the regression operator is evaluated. Under this perspective, weighted sums over discretized observations are exact evaluations of the operator under a discrete measure. This viewpoint provides a unified interpretation of functional regression, multivariate regression, and related approaches.

The theoretical results establish three key points. First, the standard regression taxonomy can be expressed as a single operator under different measure choices. Second, discrete representations correspond to exact operator evaluations and converge to the continuous operator as the observation grid is refined. Third, estimation under the discrete-measure formulation reduces to standard multivariate regression, with statistical properties governed by classical results.

The simulation study illustrates these properties empirically. In particular, it shows that discretization error behaves as predicted by numerical integration theory, that estimation error follows standard $n^{-1/2}$ scaling, and that differences in performance arise from conditioning and regularization rather than from differences in model structure. An important implication is that accurate operator estimation does not require precise recovery of the full coefficient surface, highlighting the distinction between coefficient estimation and operator evaluation.

Overall, this work suggests that many commonly used regression methods can be understood as estimating the same underlying operator under different measure specifications. This perspective clarifies the relationship between functional and multivariate regression and provides a meaningful framework for interpreting the effects of discretization, conditioning, and estimation strategy.

Several directions for future work remain. These include extending the framework to nonlinear operator models, developing inference procedures under general measures, and exploring applications to more complex data structures. More broadly, the measure-theoretic viewpoint may provide a useful lens for understanding a wider class of statistical learning methods that operate on discretized representations of continuous objects.

\appendix
\section{Technical Details for Simulation Study}

This appendix provides a formal description of the data-generating process, discretization, estimation procedures, and error metrics used in the simulation study. The presentation is organized to reflect the operator–measure framework developed in the main text.

\subsection{Operator Formulation}

Let $S = T = [0,1]$, and let $\beta : S \times T \to \mathbb{R}$ define a linear operator
\[
(T_\beta X)(t) = \int_0^1 \beta(s,t)\, X(s)\, ds,
\]
mapping functions in $L^2([0,1])$ to functions in $L^2([0,1])$.

The simulation treats regression as an operator estimation problem: the goal is to estimate the mapping $T_\beta$, rather than the coefficient surface $\beta(s,t)$ in isolation.

\subsection{Discretization as a Measure}

Let $\{s_i\}_{i=1}^p$ be a grid on $[0,1]$. We define a discrete measure
\[
\mu = \sum_{i=1}^p w_i \delta_{s_i}.
\]

In this simulation, we take
\[
s_i = \frac{i - 0.5}{p}, \qquad w_i = \frac{1}{p},
\]
so that
\[
\mu = \frac{1}{p} \sum_{i=1}^p \delta_{s_i},
\]
which is the uniform discrete probability measure on $[0,1]$.

\paragraph{Why $w_i = 1/p$.}
The choice $w_i = 1/p$ reflects the modeling assumption that each grid point represents an equal portion of the domain. Under this interpretation, the integral operator
\[
(T_\beta^\mu X)(t) = \sum_{i=1}^p w_i \beta(s_i,t)\, X(s_i)
\]
is the exact evaluation of the operator with respect to the discrete measure $\mu$.

It is important to distinguish this from classical numerical integration. In standard quadrature, weights such as $1/p$ arise as approximations to the Lebesgue integral. In contrast, here the weights define the measure itself. The operator $T_\beta^\mu$ is not an approximation of $T_\beta$ under Lebesgue measure; it is the exact operator under the discrete measure $\mu$.

\paragraph{Continuous vs discrete operators.}
The continuous operator is
\[
(T_\beta X)(t) = \int_0^1 \beta(s,t)\, X(s)\, ds,
\]
while the discrete operator is
\[
(T_\beta^\mu X)(t) = \sum_{i=1}^p w_i \beta(s_i,t)\, X(s_i).
\]

Approximation arises only when comparing $T_\beta^\mu$ to $T_\beta$. As the grid becomes dense, $\mu$ provides an increasingly fine discretization of the domain, and $T_\beta^\mu$ approaches $T_\beta$.

\paragraph{Generalization.}
The framework allows arbitrary discrete measures
\[
\mu = \sum_{i=1}^p w_i \delta_{s_i}, \qquad w_i > 0.
\]
Different choices of $(s_i, w_i)$ correspond to different models. Uniform midpoint weights are used here for simplicity and to produce well-understood convergence behavior.

\subsection{Data-Generating Process}

The data-generating process is designed to produce smooth functional predictors with a controlled covariance structure. This allows the simulation to separate the effects of discretization, finite-sample estimation, and conditioning while keeping the true operator known.

Functional predictors are generated via a truncated Karhunen--Lo\`eve expansion
\[
X(s)
=
\sum_{k=1}^K \xi_k \phi_k(s),
\]
where \(\{\phi_k\}_{k=1}^K\) are orthonormal basis functions on \([0,1]\). In the simulations, these basis functions are chosen to be smooth sinusoidal functions so that the resulting trajectories resemble typical smooth functional observations. The random coefficients are generated independently as
\[
\xi_k \sim \mathcal{N}(0,\lambda_k),
\qquad
\lambda_k=k^{-2}.
\]
The decreasing eigenvalues \(\lambda_k=k^{-2}\) ensure that lower-order basis functions explain most of the variability, while higher-order components contribute progressively less. This produces smooth functional predictors and induces the strong correlation among nearby grid values that is typical of functional data.

The response is generated according to
\[
Y(t)
=
(T_\beta X)(t)+\varepsilon(t),
\]
with \(\varepsilon(t)\) independent Gaussian noise.
Functional predictors are generated via a truncated Karhunen--Lo\`eve expansion
\[
X(s) = \sum_{k=1}^K \xi_k \phi_k(s),
\]
where $\{\phi_k\}$ are orthonormal basis functions and $\xi_k \sim \mathcal{N}(0,\lambda_k)$ with $\lambda_k = k^{-2}$.

\subsection{Vectorized Regression and Scaling}

Let $X_k(s)$ denote the $k$-th sample. Define the design matrix
\[
X_{\text{mat}}[k,i] = \sqrt{w_i}\, X_k(s_i).
\]

\paragraph{Role of the scaling.}
The factor $\sqrt{w_i}$ ensures that the empirical Gram matrix
\[
\frac{1}{n} X_{\text{mat}}^\top X_{\text{mat}}
\]
approximates the population Gram operator under the discrete $L^2(\mu)$ inner product
\[
\langle f, g \rangle_{L^2(\mu)} = \sum_i w_i f(s_i) g(s_i).
\]

Without this scaling, the regression problem would be improperly weighted and would not correspond to the operator under $\mu$.

\subsection{Recovery of the Coefficient Surface}

The regression model is
\[
Y = X_{\text{mat}} B + E.
\]

Because the predictors are scaled by $\sqrt{w_i}$, the estimated coefficient matrix $\hat{B}$ corresponds to a scaled version of the kernel. The coefficient surface is recovered via
\[
\hat{\beta}(s_i,t_j) = \frac{1}{\sqrt{w_i}} \hat{B}_{ij}.
\]

For the uniform measure $w_i = 1/p$, this reduces to
\[
\hat{\beta}(s_i,t_j) = \sqrt{p}\, \hat{B}_{ij}.
\]

This scaling ensures that $\hat{\beta}$ is expressed on the same scale as the true coefficient surface.

\subsection{Error Metrics}

We define three distinct error measures.

\subsubsection{Coefficient Error (ISE)}

\[
\text{ISE} = \sum_{i,j} w_i w_j \left( \hat{\beta}(s_i,t_j) - \beta(s_i,t_j) \right)^2.
\]

This measures recovery of the coefficient surface in the discrete $L^2(\mu \times \mu)$ norm.

\subsubsection{Operator Error (OpErr)}

Define noiseless outputs
\[
Y_{\text{signal}}(t_j) = \sum_i w_i \beta(s_i,t_j) X(s_i),
\]
\[
\hat{Y}_{\text{signal}}(t_j) = \sum_i w_i \hat{\beta}(s_i,t_j) X(s_i).
\]

Then
\[
\text{OpErr}^2 = \frac{1}{n_{\text{test}}} \sum_k \sum_j w_j
\left( \hat{Y}_{\text{signal},k}(t_j) - Y_{\text{signal},k}(t_j) \right)^2.
\]

This measures the accuracy of the operator action.

\paragraph{Interpretation.}
ISE measures accuracy of $\beta$, whereas OpErr measures accuracy of $T_\beta$. These quantities need not behave similarly, since the operator depends only on components of $\beta$ aligned with the input functions.

\subsubsection{Discretization Error (DiscErr)}

Using a high-resolution reference grid,
\[
\text{DiscErr}^2 =
\frac{1}{n_{\text{test}}} \sum_k \sum_j \frac{1}{p_{\text{ref}}}
\left( (T_\beta^{\mu_{\text{ref}}} X_k)(t_j) - (T_\beta^\mu X_k)(t_j) \right)^2.
\]

This isolates the difference between discrete and continuous operator representations.

\subsubsection{Prediction Error (RMSE)}

\[
\text{RMSE}^2 = \frac{1}{n_{\text{test}} p} \sum_{k,j}
\left( Y_k(t_j) - \hat{Y}_k(t_j) \right)^2.
\]

Taking expectations,
\[
\mathbb{E}[\text{RMSE}^2] \approx \text{OpErr}^2 + \sigma^2.
\]

\paragraph{Interpretation.}
RMSE reflects both estimation error and irreducible noise, whereas OpErr isolates the estimation component.

\subsection{Summary}

The simulation distinguishes three sources of error:

\begin{itemize}
\item discretization error (choice of measure),
\item estimation error (finite samples and conditioning),
\item observation noise.
\end{itemize}

This separation is central to the interpretation of results in Section 6.

\end{document}